\documentclass[]{interact}

\usepackage{epstopdf}

\usepackage[numbers,sort&compress]{natbib}
\bibpunct[, ]{[}{]}{,}{n}{,}{,}
\makeatletter
\def\NAT@def@citea{\def\@citea{\NAT@separator}}
\makeatother

\theoremstyle{plain}
\newtheorem{theorem}{Theorem}[section]

\newtheorem{problem}[theorem]{Problem}

\theoremstyle{definition}

\theoremstyle{remark}

\usepackage[linesnumbered,ruled,algo2e]{algorithm2e}
\usepackage{subfigure}
\usepackage{graphicx}
\usepackage{grffile}

\begin{document}


\title{A Geo-Aware Server Assignment Problem for Mobile Edge Computing}

\author{
\name{Duc A. Tran\thanks{CONTACT Duc A. Tran. Email: duc.tran@umb.edu} and Quynh Vo}
\affil{Department of Computer Science, University of Massachusetts, Boston, MA 02125}
}

\maketitle

\begin{abstract}
As mobile devices have become the preferred tool for communication, work, and entertainment,  traffic at the edge of the network is growing more rapidly than ever. To improve user experience, commodity servers are deployed in the edge to form a decentralized network of mini datacenters each serving a localized region. A challenge  is how to place these servers geographically to maximize the offloading benefit  and be close to the users they respectively serve.  We introduce a formulation for this problem to serve applications that involve pairwise communication between mobile devices at different geolocations. We explore several heuristic solutions and compare them in an evaluation using both real-world and synthetic datasets. 
\end{abstract}


\section{Introduction \label{sec:intro}}
Mobile Edge Computing (MEC) \cite{MECWhitePaper} has emerged as a viable technology for mobile operators to push computing resources closer to the users so that requests  can be served locally without long-haul crossing of the network core, thus improving network efficiency and user experience. In a nutshell, a typical MEC architecture consists of four layers of entities: the mobile users, the base-stations, the edge servers, and the cloud datacenter. The edge servers are introduced in the edge of the network connecting the base-stations to the network core, each server being an aggregation hub, or a mini datacenter to offload processing tasks from the remote datacenter. Because the region, or ``cloudlet" \cite{5280678}, served by an edge server is much smaller, a commodity virtualization-enabled computer can be used to run compute and network functions that would otherwise be provided by the datacenter. 

MEC can benefit many compute-hungry or latency-critical applications involving video optimization \cite{8010284}, content delivery \cite{7979975}, big data analytics \cite{7994559}, and augmented reality \cite{7906521}, to name a few.  Originally initiated for cellular networks to realize the 5G vision, MEC has been generalized for broader wireless and mobile networks \cite{8016573}. It is  becoming more of a phenomenon with the Internet of Things; more than 5 billion IoT devices would be connected to MEC  by 2020 according to a January 2017 forecast by BI Intelligence \cite{BIIntelligenceIoT}. 

 A challenge with MEC is how to align the edge servers with the base-stations geographically to maximize the edge computing benefits. To address this challenge is application-specific. We focus on applications involving pairwise transactions between devices. Cellphone calls made from one user to another, peer-to-peer video streaming, and multi-player online gaming are examples of this type of communication. If the two devices are served by different edge servers, the datacenter must get involved, thus incurring a backhaul  cost. This cost is avoided if the same  server serves both  devices. However, it makes no practical sense if they are located in far remote geographic locations because a server should be  geographically close to where it serves to avoid high installation cost and long latency \cite{Bouet:2017:GMR:3098208.3098216,7367390}. On the other hand, those devices with many transactions should   belong to the same  server. Although geographic proximity tends to imply high transactional activity, this relationship is not straightforward. As the number of servers is finite and their capacities limited, it is impossible to equally please all the users.

Therefore, we are motivated to solve the following optimization problem: where to place a set of edge servers of limited capacity and assign them to  the base-stations such that (1) offloading benefit is maximized and (2) each server is geographically close to its respective users. We require that the server locations be chosen from a set of predetermined geographic sites; this constraint applies to practical cases involving non-technical factors  in the deployment of the servers, for example, due to economics, policies, or management. 

The server assignment problem is not new outside MEC. Indeed, it belongs to the body of work on distributed allocation of resources (virtual machines) widely studied in the area of cloud computing \cite{6847919,6195847,Mann:2015:AVM:2808687.2797211,DENG2017107}.  The MEC problem is  similar, however with unique constraints. Firstly, the servers in MEC need be near the user side, not  the datacenter side, and so the communication cost to be minimized is due to the use of the backhaul network (towards the datacenter), not  the front-haul (towards the cells). Secondly, the geographic spread of the cells served by a MEC server should be a design factor, which is not a typical priority for a  distributed cloud solution. 

The MEC server assignment problem has been addressed in some forms, only recently  \cite{Ceselli:2017:MEC:3148626.3148659,Bouet:2017:GMR:3098208.3098216}. Our key contribution is a new practical formulation for the server assignment problem. We prove its NP-hardness  and subsequently explore an approximation solution based on local search heuristics. We evaluate its effectiveness and efficiency using both real-world and synthetic datasets and comparing to intuitive approaches.

The remainder of the paper is organized as follows. Related work is discussed in Section \ref{sec:relatedwork}. The problem  is stated and formulated in Section \ref{sec:prob}. The algorithm is proposed in Section \ref{sec:solution}. The results of the evaluation study are analyzed in Section \ref{sec:evaluation}.  The paper concludes in Section \ref{sec:conclusions} with pointers to our future work.

\section{Related Work\label{sec:relatedwork}}
Every finite computing system serving a large number of resource-hungry requests faces the challenge of how to assign resources to computing units to optimize hardware consumption and best satisfy application QoS requirements. The MEC server assignment problem shares the same challenge, which can arise in various scenarios. 

The assignment problem in \cite{8071527} applies to a MEC network supporting multiple applications of known request load and latency expectation  and the challenge is  to determine which edge servers to run the required virtual machines (VMs), constrained by server capacity and inter-server communication delay. In \cite{DBLP:conf/edge/BahreiniG17,Wang:2017:10.1109/ACCESS.2017.2665971}, where only one application is being considered and consists of multiple inter-related components organizable into a graph, the challenge is  how to place this component graph on top of the physical graph of edge servers to minimize the cost to run the application. In the case that edge servers must be bound to certain geographic locations, a challenge is to decide which among these locations  we should place the servers and   inter-connect them for optimal routing and installation costs \cite{Ceselli:2017:MEC:3148626.3148659}.

The above works do not take into account the geographic spread of the region served by a server. The cells served by the same server can be highly scattered geographically,  causing long latency and high management cost. This motivates the work in \cite{Bouet:2017:GMR:3098208.3098216}  proposing a spatial partitioning of the geographic area such that  the cells of the same server are always contiguous. For this partitioning, a graph-based clustering algorithm is proposed that repeatedly merges adjacent cells to form clusters as long as the merger results in better offloading and no cluster exceeds the server capacity. In a similar research \cite{7367390}, where the cells served by each server are also contiguous, the objectives are to minimize the server deployment cost, the front-haul link cost for each server to reach its assigned base-stations, and the cell-to-cell latency  via the edge; the proposed algorithm is to repeatedly select the next remaining server of the least deployment cost and assign it to all the nearby base-stations of the least front-haul link cost so long as the server capacity is met. 

The problem in \cite{Bouet:2017:GMR:3098208.3098216}  is aimed to optimize for workloads involving cell-to-cell communication, whereas the problem in  \cite{7367390}  is for individual-cell workloads.  The latter also  requires that all the processing be fulfilled by the edge servers, hence zero backhaul cost. In this aspect, our work is  more similar to \cite{Bouet:2017:GMR:3098208.3098216} because we also optimize for cell-to-cell workloads and cannot avoid backhaul use (thus, the objective to minimize its cost). 
However, there are   key differences. 
First, the number of servers is a constraint in our problem, but not in \cite{Bouet:2017:GMR:3098208.3098216}.  Second, we   minimize the geographic spread of the cells served by each server, instead of enforcing their geographic contiguity. We argue that these cells should not be too far from their server but do not have to be contiguous; in contrast, the solution in \cite{Bouet:2017:GMR:3098208.3098216} enforces contiguity, but has no control of spread. Third,  we require that the servers be bound to predetermined locations (as in \cite{7367390,Ceselli:2017:MEC:3148626.3148659}, but not a constraint of \cite{Bouet:2017:GMR:3098208.3098216}).

\section{Problem Statement\label{sec:prob}}
 The  geographic area  $\mathcal{A}$ is partitioned into a set $\mathcal{C}$ of $N_\mathcal{C}$  cells, each cell $i\in \mathcal{C}$ served by a base-station at a known location in $\mathcal{A}$; we also refer to this base-station as base-station $i$. The meanings of ``base-station" and ``cell" are general, not necessarily understood in conventional meaning as in cellular networks; for example, as a Wi-Fi access router and its coverage area in Wi-Fi networks.
The edge layer consists of $N_\mathcal{S} < N_\mathcal{C}$ edge servers, whose  locations  are chosen from a set $\mathcal{L}  \subset \mathcal{A}$ of $N_\mathcal{L}   \ge N_\mathcal{S}$ candidate locations. For example,   $\mathcal{L}$ can be a subset of  base-station locations; in this case,  an edge server must be co-located with some base-station (as in \cite{7367390}). In general, we admit any arbitrary candidate location in $\mathcal{A}$.

The  input workload  is a symmetric $N_\mathcal{C} \times N_\mathcal{C}$ matrix of non-negative real values $w_{ij}$ representing the transaction demand between users in cell $i \in \mathcal{C}$ with users in cell $j \in \mathcal{C}$. Note that $w_{ii}$ is the transaction demand between users of the same cell $i$. Denote by $w_i = \sum_{j \in C} w_{ij}$ the total workload involving cell $i$. Without loss of generality, assume that the total workload involving all the cells equals 1; i.e., $\sum_{i \le j \in \mathcal{C}} w_{ij} =1$.     Each server exclusively manages the workload for a group of base-stations. If  base-station $i$ is assigned to a server at location $l$, we have a front-haul link, whose usage cost should increase with their distance; denote this cost by $d_{il}$ (assumed given, e.g., equal to distance).  Because the backhaul links to reach the remote data center are much more expensive, representing the worst-case scenario, we assume their cost to be a fixed cost $d$ much higher than all front-haul link costs. We want to avoid  backhaul links as much as possible.

Our goal is to (1) server location assignment (SLA): assign the best location for each server and (2) cell-server assignment (CSA): assign the best server for each cell. We use 0-1 integer programming to formulate these assignments. Define a binary variable $z = (z_{il})_{\mathcal{C}\times \mathcal{L}}$ such that $z_{il}=1$ iff there is a server at location $l$ and this server serves cell $i$. 
Given $z$, we can tell exactly the server locations (SLA) and the cell-to-server assignment (CSA). A location $l$ is a server location iff $\sum_{i \in \mathcal{C}} z_{il} \ge 1$,
i.e., at least one cell is assigned to location $l$. Another way to express this condition is
\[
\prod_{i \in \mathcal{C}} (1-z_{il}) = 0.
\]
Because   $N_\mathcal{S}$ different locations must be chosen for the servers, we have the constraint below,
\[
\sum_{l \in \mathcal{L}}\bigg (1- \prod_{i \in \mathcal{C}} (1-z_{il}) \bigg) = N_\mathcal{S}.
\]
Also, each cell must   be assigned to exactly one server, another constraint is
\[
\sum_{l \in \mathcal{L}} z_{il} = 1 ~\forall i \in \mathcal{C}.
\]
We assume 
that there is a capacity $W<1$ on  the compute load a server can process. In the case a server is fully saturated, the residual workload must be serviced by the data center.  Our objectives are   to minimize the backhaul cost and geo-spread under this    assumption. 

\subsection{Backhaul Cost}
A transaction can be one of the following types: between users of the same cell, between users of two different cells assigned to the same server, and between users of two different cells assigned to different servers.
The compute demand for the edge comes from transactions of the  first two types. Specifically, if a server is placed at location $l$, we represent its compute (demand) load as 
\begin{align}
W(l) &= \sum_{i \le j  \in \mathcal{C}} z_{il} z_{jl} w_{ij}.
\label{eq:frontload}
\end{align}
Of course, $W(l)=0$ for every non-server location $l$.

If the server capacity is infinite, all of this load will be fulfilled by the server. However, limited by the server capacity $W$, if $W(l) > W$, the remaining amount ($W(l)-W$) of workload must be processed at the datacenter, thus incurring a backhaul cost.
Consequently, the total backhaul cost is due to not only the transactions between cells assigned to different servers, but also those transactions assigned to the same server that exceed its capacity. We represent the backhaul cost as
\begin{align}
COST = \sum_{i\le j \in \mathcal{C}} w_{ij} \underbrace{\prod_{l \in \mathcal{L}} (1-z_{il}z_{jl})}_{ i,j ~ not~ assigned~ to ~ same ~server} 
+ \sum_{l \in \mathcal{L}} \underbrace{\max \bigg(0,  \sum_{i \le j  \in \mathcal{C}} z_{il} z_{jl} w_{ij}-W \bigg)}_{overload ~amount ~ of ~ server ~ at ~l},
\label{eq:o_cost}
\end{align}
which can also be expressed as
\begin{align}
COST &= \bigg(1 - \sum_{s \in \mathcal{S}}  \sum_{i\le j \in \mathcal{C}} w_{ij}x_{is}x_{js} \bigg) 
+ \sum_{s \in \mathcal{S}} \max \bigg(0,  \sum_{i \le j  \in \mathcal{C}} x_{is} x_{js} w_{ij}-W \bigg)\nonumber\\
=& 1 + \sum_{s \in \mathcal{S}} \max \bigg(-W, -\sum_{i\le j \in \mathcal{C}} w_{ij}x_{is}x_{js}\bigg)\nonumber\\
=& 1 - \sum_{s \in \mathcal{S}} \min \bigg(W, \sum_{i\le j \in \mathcal{C}} w_{ij}x_{is}x_{js}\bigg).
\label{eq:o_cost1}
\end{align}
We want to minimize $COST$.

\subsection{Geographic Spread}
For better latency and easier management, we should keep the geographic region served by an edge server from spreading too far, especially for cells with many transactions. We quantify the geo-spread of a server as the sum of its distance to each assigned base-station, weighted by transaction demand. If this server is placed at location $l$, its geo-spread is quantified as
\begin{align}
S(l) = \sum_{i \in \mathcal{C}} z_{il} d_{il}  w_{i}.
\end{align}
We want to minimize the total geo-spread for all the  servers 
\begin{align}
   SPREAD 
= \sum_{l \in \mathcal{L}} S(l) = \sum_{l \in \mathcal{L}}\sum_{i \in \mathcal{C}} z_{il} d_{il}  w_{i}   .
\label{eq:o_spread}
\end{align} 
\subsection{Optimization Problem and NP-Hardness}
In summary, our problem is the following two-objective optimization problem.
\begin{problem}[Min Cost-Spread Assignment (MCSA)]
\begin{align}
\min ~ &\{COST, SPREAD \} \nonumber\\
\text{s.t.}~~~& \nonumber \\
~~~&1)~\sum_{l \in \mathcal{L}}\bigg (1- \prod_{i \in \mathcal{C}} (1-z_{il}) \bigg) = N_\mathcal{S}.\label{c00}\\
~~~&2)~\sum_{l \in \mathcal{L}} z_{il} = 1 ~\forall i \in \mathcal{C} \label{c01}\\
~~~&3)~z_{il} \in \{0,1\} ~\forall i\in \mathcal{C},l\in \mathcal{L}.
\end{align}
\label{ref:problem}
\end{problem}
\begin{theorem}
  MCSA     is NP-hard.
\end{theorem}
\begin{proof}
Consider a simple configuration of our problem: $w_{ij}=0$ for all $(i,j) \in \mathcal{C}\times \mathcal{C}$ except for $i=j$, and $d_{il}=1$ for all $(i,l) \in \mathcal{C}\times \mathcal{L}$. Then, it is easy to see that 
 \begin{align*}
COST&= \sum_{l \in \mathcal{L}}  \max  (0,  \sum_{i    \in \mathcal{C}} z_{il}  w_{ii}-W )\\
SPREAD&=\sum_{l \in \mathcal{L}}\sum_{i \in \mathcal{C}} z_{il} w_{ii} 
= \sum_{i \in \mathcal{C}} w_{ii}\sum_{l \in \mathcal{L}}z_{il} =  \sum_{i \in \mathcal{C}} w_{ii}.
\end{align*}
Because $SPREAD$ is a constant, we can choose any $N_\mathcal{S}$ locations to place the servers. Given these server locations, what remains is to minimize $COST$,
\[
\min \bigg \{ COST = \sum_{l \in \mathcal{L}}  \max  (0,  \sum_{i    \in \mathcal{C}} z_{il}  w_{ii}-W ) \bigg\}
\]

This minimization is NP-hard because we show below that an algorithm for it can be used to solve the optimization version of the partition problem, which is known to be NP-hard: partition a given set of positive integers, $\{x_1, x_2, ..., x_n\}$, into two subsets such that the respective subset sums differ the least. To reduce to our problem, consider $N_\mathcal{S}=2$ servers  and  $N_\mathcal{C}=n$ cells $\{1,2,...,n\}$  with $w_{ii} = x_i$, and let $W= \frac{1}{2}\sum_{i=1}^n w_{ii}$. Suppose that an optimal $COST$ solution, $(A,B)$, assigns a cellset $A \subset [n]$   to  server 1 and a cellset $B \subset [n]$  to server 2. Without loss of generality, let $\sum_{i \in A} w_{ii} \le \sum_{i \in B} w_{ii}$. The corresponding $COST$ is
\begin{align*}
COST(A,B) 
= \max(0, \sum_{i \in A} w_{ii}-W) + \max(0, \sum_{i \in B} w_{ii}-W)
=    \sum_{i \in B} w_{ii}-W.
\end{align*}
There are two cases. First, if  $\sum_{i \in A} w_{ii} = \sum_{i \in B} w_{ii} = W$, then partition $A \cup B$ is optimal for the partition problem because the subset sums are identical. Second, in the otherwise case,  $\sum_{i \in A} w_{ii} < W < \sum_{i \in B} w_{ii} $, partition $A \cup B$ is optimal for the partition problem because  no partition $A' \cup B'$ can offer a smaller subset sum difference,
\[
 \sum_{i \in B'} w_{ii} - \sum_{i \in A'} w_{ii}  <  \sum_{i \in B} w_{ii} - \sum_{i \in A} w_{ii}. 
 \]
 Indeed,  suppose by contradiction that this partition $(A', B')$ exists. Without loss of generality, let $\sum_{i \in B'} w_{ii} \ge \sum_{i \in A'} w_{ii}$; hence, we must have $\sum_{i \in B} w_{ii} > \sum_{i \in B'} w_{ii} \ge W \ge \sum_{i \in A'} w_{ii}$. Then, if we assign $A'$ to server 1 and $B'$ to server 2, $COST$ will be
\begin{align*}
COST(A',B') &= \max(0, \sum_{i \in A'} w_{ii}-W) + \max(0, \sum_{i \in B'} w_{ii}-W)\\
&= 0 +  \sum_{i \in B'} w_{ii}-W<\sum_{i \in B} w_{ii} -W\\
&= COST(A,B).
\end{align*}
This is contradictory to the assumption that $(A,B)$ is the optimal $COST$ solution.
\end{proof}

\section{Heuristic Approach\label{sec:solution}}
We propose a three-phase algorithm approach: focus on spread optimization first, then refine the solution based on the cost, which as a side effect may worsen the spread, and, finally, improve the solution again, this time for a better spread.  As local search is widely used for hard combinatorial optimization problems, we present below the local search methods to optimize each individual objective and  how to apply them in the three-phase approach.

\subsection{Cost-Only Optimization}
Because $COST$ does not involve geography, we   need   compute only the best cell-server assignment based on the  workload demand; any location choice for the servers would work. In a nutshell,  our algorithm starts with a random assignment (random cell-server assignment and random server-location assignment), and repeatedly apply a local operation such that the new assignment improves $COST$. A local operation, denoted by $\mathsf{move\_cells}(l_0, l_1)$, runs an algorithm to migrate cells  between a pair of  servers at locations $l_0$ and $l_1$; the servers are referred to as the $l_0$-server and $l_1$-server, respectively. As long as we can find a local operation that improves $COST$, 
\begin{align}
COST_{new} < COST,
\label{eq:movecells1}
\end{align}
we make the new assignment permanent and repeat the same process until no such local operation is found. 

Let $A_0 (A_0')$ and $A_1 (A_1')$ denote the cellsets of the $l_0$-server and $l_1$-server before (after) the location, respectively. According to Eq. (\ref{eq:o_cost}), $COST$ will decrease if the   quantity
\begin{align}
\Delta_{COST}(A_0, A_1) = \sum_{i \in \mathcal{A}_0}\sum_{j \in \mathcal{A}_1} w_{ij} + \max \bigg(0,  \sum_{i \le j  \in \mathcal{A}_0}  w_{ij}-W \bigg) + \max \bigg(0,  \sum_{i \le j  \in \mathcal{A}_1}  w_{ij}-W \bigg)
\label{eq:bipartition}
\end{align}
decreases as a result of replacing  $(A_0 , A_1)$ with   $(A_0' , A_1')$. Consequently, we should design the cell-moving algorithm  such that its objective is to minimize   $\Delta_{COST}$. 

This challenge can be translated into a graph bipartitioning problem. Let $H$  be a weighted graph (self-loop possible) where each cell in $C$ is a vertex and an edge connects  cell $i$ and cell $j$ if they have transactions; the weight of edge $(i,j)$ is $w_{ij}$. A feasible solution is a partition of $H$ into two components. The first additive term of Eq. (\ref{eq:bipartition}) is the cut weight of this partition and the second and third additive terms represent a capacity-constrained quality for the partition. We derive an algorithm to compute the best partition  based on the Fiduccia-Mattheyses (FM) heuristic  \cite{Fiduccia:1982:LHI:800263.809204}. FM is  effective  for solving the classic graph min-cut bipartitioning problem whose objective is to minimize the cut weight while balancing the vertex weight. FM is fast (linear time in terms of the number of vertices) and simple (each local operation involves moving only one vertex  across the cut). Because our objective is different (minimizing $\Delta_{COST}$), we need to modify FM. 

\begin{algorithm2e}[t]
 $A_0, A_1$: cellsets of $l_0$-server and $l_1$-server, respectively\;
Compute $W(A_0),W(A_1)$: edge weight sum (compute load) of  $A_0$ and  $A_1$, respectively\;
\While {true} {
	$k^*=0$, 	$GAIN_{max}=0$, 	$GAIN=0$\;
	Unlock all vertices\;
	Compute gain for every vertex\;
	Save the current partition $C=A_0 \cup A_1 $\;
	\For{\texttt{$k=1, 2, ..., |C|$}} {
     	   	\If {$\exists$ an unlocked vertex $i_k \in C$ of highest gain  satisfying InEq. (\ref{eq:maxloadcondition}) } {
			  $GAIN ~+=~ gain(i_k)$\;
        			\If {$GAIN > GAIN_{max}$} {
        				$GAIN_{max} = GAIN$\;
        				$k^* = k$\;
        			}	
			Move vertex $i_k$  to the other component\;
		 	Update gain for every unlocked vertex\;
			Update $W(A_0)$ and $W(A_1)$\;	
			Lock vertex $i_k$\; 
		}
		\lElse {
			\textbf{break}
		}
       		
        }
        \If {$GAIN_{max} > 0$} {
        		Retrieve the original  current partition   $ C=A_0 \cup A_1 $\;
        		Move vertices $i_1, i_2, ..., i_{k^*}$ each from its respective original component to the other component\;
		Update gain for every unlocked vertex\;
		Update $W(A_0)$ and $W(A_1)$\;
	 }
        \lElse {
	        \textbf{break}
        }
}
\Return\;
\caption{$\mathsf{move\_cells}(l_0,l_1)$} 
\label{alg:psi}
\end{algorithm2e}

The cell-moving algorithm works as follows (see Algorithm \ref{alg:psi}). Resembling FM, the algorithm runs in passes and in each pass we compute a sequence of cell migrations each  moving a vertex from $A_0$ to $A_1$   or from $A_1$ to $A_0$ such that $\Delta_{COST}(A_0,A_1)$ after this series is maximally improved. The algorithm stops when no improvement can be made. To determine which vertex to move, let us define for each vertex a quantity called   ``gain", which is the cost reduction if the vertex were moved from its component to the other component.  Consider a vertex $i$ and, without loss of generality, suppose that $i \in A_0$. If vertex $i$ were moved to $A_1$, its gain in the cut weight is  
\[
gain_{cut}(i) =  \sum_{j \in A_1} w_{ij} -  \sum_{j \in A_0, j\neq i} w_{ij}.
 \]
The   edge weight sum of $A_0$ and that of $A_1$ would be changed to
 \begin{align*}
 W'(A_0) & = W(A_0) -  \sum_{j \in A_0} w_{ij} \\
 W'(A_1) &= W(A_1) + w_{ii} + \sum_{j \in A_1} w_{ij}
 \end{align*}
 and so the gain in the capacity-constrained quality is
 \begin{align*}
 gain_{cap}(i)=\max(0, W(A_0)-W) + \max(0, W(A_1)-W)\\
  - \max(0, W'(A_0)-W) - \max(0, W'(A_1)-W).
 \end{align*}
 The gain of vertex $i$ is
 \[
 gain(i) = gain_{cut}(i)+gain_{cap}(i).
 \]
Intuitively, a positive (negative) gain would result in a smaller (larger) $\Delta_{COST}$ if the vertex switched its component.
   
    At the beginning of each pass, we construct a priority queue of vertices based on their gain. This queue includes only those vertices  whose migration would result in
 \begin{align}
 \max(W'(A_0), W'(A_1)) \le  \max(W,W(A_0), W(A_1)).
 \label{eq:maxloadcondition}
 \end{align}
 In other words, we  consider moving a vertex only if the moving improves the load balancing between the two servers or keeps them under the server capacity.
 
During the current pass, we repeatedly select the vertex of highest gain from the priority queue,  move it, and update the queue. After this vertex is moved, it is ``locked"  so that it cannot be moved again in the current pass. Then, we repeat this process until the queue is empty. We keep track of the gain accumulation   after each $k^{th}$ step:
\[
GAIN_k = \sum_{t=1}^k gain(i_t),
\]
where $i_t$ is the vertex chosen in step $1 \le t \le k$.
The best move decision  would be to move vertices $i_1, i_2, ..., i_{k^*}$ such that $k^* = \arg \max_{k \le N_\mathcal{C}} GAIN_k$. 

If  $GAIN_{k^*} > 0$, these moves would result in better $\Delta_{COST}$ because the new $\Delta_{COST}$ is  $\Delta_{COST} - GAIN_{k^*} < \Delta_{COST}$; we make these moves permanent  and   go on to the next pass which repeats the same procedure. Else, the   algorithm makes no change (i.e., keep the same partition as that before the pass starts) and stops.  

When a local operation $\mathsf{move\_cells}(l_0,l_1)$ finishes, we will use the final assignment resulted from this operation. The algorithm continues repeatedly with finding the next $\mathsf{move\_cells}()$ local operation that can further improve $COST$ and stops when  no such local operation is found.

\subsection{ Spread-Only Optimization}
To minimize $SPREAD$ can be reduced to solving  a  k-median  problem \cite{doi:10.1137/0137041}.
In k-median, given a set of $n$ clients and a set of $m$ facilities, the goal is to choose $k$ facilities to open and assign an open facility to each client such that the total assignment cost is minimum, assuming that  the assignment cost to service client $i$ by facility $l$ is $\gamma_{il}$ (by default, a metric).  
We can consider each server a facility to open ($k= N_\mathcal{S}$, $m  = N_\mathcal{L}$) and each cell a client ($n= N_\mathcal{C}$). The cost to assign client $i$ to facility $l$ (if open) is  $\gamma_{il} = w_id_{il}$ (alternatively, we can think of $d_{il}$ as the assignment cost per unit of service and $w_i$ as the service demand).  Then  the total service  cost of the corresponding k-median problem is 
\[
\sum_{l \in \mathcal{L}: ~open} \sum_{i \in \mathcal{C}} z_{il} \gamma_{il}
=
\sum_{l \in \mathcal{L}}\sum_{i \in \mathcal{C}} z_{il} d_{il}  w_{i}
= SPREAD.
\]
Therefore,  any   k-median solution $z_{il}$  ($z_{il}=1$ iff facility $l$ is open and client $i$ is served by this facility) is   a solution that minimizes $SPREAD$ for our problem.  

K-median is NP-hard \cite{doi:10.1137/0137041} and the best  approximation factor known to date is $2.675+\epsilon$, achieved by Byrka et al. \cite{Byrka:2017:IAK:3040971.2981561}. Using the local search approach,  one can obtain an approximation factor of $(3+2/p)$, for example by Arya et al.'s polynomial-time algorithm  \cite{Arya:2001:LSH:380752.380755}. This algorithm starts with a feasible assignment and repeatedly perform a $p$-facility swap until no further cost reduction.

Similarly, our algorithm starts with a random assignment  and then repeatedly applies a series of local operations. Let $\mathsf{assign\_cells}(L)$ denote an algorithm that assigns the cells in $\mathcal{C}$ to the servers located at a given subset of locations, $L \subset \mathcal{L}$, such that a cell is always assigned to the nearest server; i.e.,  
\[
\mathsf{assign\_cells}(L): i \mapsto \arg \min_{l \in L} d_{il}.
\]
A local operation, denoted by $\mathsf{swap\_locations}(l_0, l_1)$,   involves a pair of a server location $l_0$ in the current server location set $L$ and a non-server location $l_1 \not \in L$,  and does the following:
\begin{itemize}
\item Remove location $l_0$ from the server set
\item Add location $l_1$ as a new server location set 
\item Run  $\mathsf{assign\_cells}(L')$ to obtain a new cell-server assignment where $L' =  L-\{l_0\}+\{l_1\}$ is the new server set. 
\end{itemize}
A local operation is chosen to take place permanently if  $SPREAD$ of the resultant assignment is improved  by at least a constant factor $\kappa \in (0,1)$; i.e.,
\begin{align}
SPREAD_{new} < (1-\kappa) \times SPREAD.
\label{eq:facilityswap}
\end{align}
Subsequently, the algorithm goes on repeatedly with finding another local operation satisfying this inequality until none is found.
\subsection{Three-Phase Algorithm}
The above algorithms are designed for  only one objective, cost or spread.  We propose the following three-phase algorithm; a summary  is given in Algorithm \ref{alg_threephase}.

In Phase 1 (lines 1-4 of Algorithm \ref{alg_threephase}), we run the spread-only algorithm presented above to obtain an assignment with the (approximately) best spread. 

In Phase 2 (lines 5-7 of Algorithm \ref{alg_threephase}), we start with this  assignment and adjust the cell-server assignment to improve cost. During the process, the server-location assignment is intact. For the adjustment, we apply the same local search algorithm (same local operation) as in the cost-only algorithm except for one small modification.  Specifically, a local operation,  $\mathsf{move\_cells}(l_0, l_1)$, is made permanent not only if the resultant cost is less (Eq. (\ref{eq:movecells1}), but also the resulted spread remains below a threshold,
\begin{align}
SPREAD_{new} \le (1+\epsilon) \times SPREAD_{0},
\label{eq:movecells2}
\end{align}
Because a local operation, while lessening the cost may worsen the spread, the  threshold is introduced to keep the spread within a reasonable factor of  $SPREAD_{0}$ that is the spread at the start of  Phase 2. Here, $\epsilon \in (0, \infty)$; we can set $\epsilon=\infty$ if the goal is to bring down the cost   aggressively.

In Phase 3 (lines 8-10 of Algorithm \ref{alg_threephase}), we start with the  assignment of Phase 2  and recompute the server locations  for better spread (which has been worsen as tradeoff during Phase 2 compared to that in Phase 1). During the process, the cell-server assignment is intact. Denote the server set by $\mathcal{S}$ and the cellset of  each server $s \in \mathcal{S}$   by $cellset(s)$. The unknown to compute is the binary variable  $y_{sl}$, set to $1$ iff server $s$ is placed at location $l$. We have 
\begin{align}
SPREAD =  \sum_{l \in \mathcal{L}}\sum_{i \in \mathcal{C}} z_{il} d_{il} w_i 
= \sum_{s\in \mathcal{S}} \sum_{l \in \mathcal{L}}y_{sl}\underbrace{\sum_{i \in cellset(s)}  d_{il} w_i}_{A_{sl}}.
\label{eq:hungarian}
\end{align}
Because a server must be assigned to an exclusive location,  to minimize $\sum_{s \in \mathcal{S}}\sum_{l \in \mathcal{L}}y_{sl} A_{sl}$ is equivalent to finding a min-cost maximal matching in a complete bipartite graph $(\mathcal{S} , \mathcal{L})$ where an edge connects a vertex $s \in \mathcal{S}$ to a vertex $l \in \mathcal{L}$ with cost $A_{sl}$ (which is known from the intact cell-server assignment). Therefore, we apply the Hungarian Algorithm \cite{Kuhn55thehungarian} to compute this matching ($y_{sl}$), which runs in polynomial time (cubic in the number of vertices).

\begin{algorithm2e}[t]
\tcc{Phase 1: K-median to minimize $SPREAD$}
Initial assignment: Choose a set $L$ of  $N_\mathcal{S}$ random locations for the servers and
assign cells to these servers randomly\;
\While {$\exists$  $l_0 \in L$ and $\exists ~l_1 \not \in L$ such that    $\mathsf{swap\_locations}(l_0, l_1)$  would result in an assignment satisfying InEq. (\ref{eq:facilityswap}) } {
Permanently apply the assignment resulted from $\mathsf{swap\_locations}(l_0, l_1)$\;
}
\tcc{Phase 2:  FM to improve $COST$}
\While {$\exists$  $l_0,l_1 \in L$   such that     $\mathsf{move\_cells}(l_0, l_1)$  would result in   an assignment satisfying InEq. (\ref{eq:movecells1}) and InEq. (\ref{eq:movecells2})} {
Permanently apply the assignment resulted from  $\mathsf{move\_cells}(l_0,l_1)$ \; 
}
\tcc{Phase 3:  Hungarian to improve $SPREAD$}
Let $\mathcal{S}$ be the set of servers (i.e., those serving cells according to the above assignment)\;
Compute matrix $[A_{sl}]_{\mathcal{S}\times \mathcal{L}}$ as defined in Eq. (\ref{eq:hungarian}) \;
Run Hungarian Algorithm on the cost matrix $[A_{sl}]_{\mathcal{S}\times \mathcal{L}}$ \;
\Return\;
\caption{Three-Phase Algorithm} 
\label{alg_threephase}
\end{algorithm2e} 

\section{Evaluation\label{sec:evaluation}}
We conducted an evaluation in two scenarios: using a synthetic dataset (Synthetic500) to represent a workload that has   no relationship with geography and a real-world dataset (Milano625) to represent a workload in which demand is higher between cells of increasing proximity. 
\begin{itemize}
\item Synthetic500:  The service area is a 2D square area $\mathcal{A} = [0,1]^2$ where $N_\mathcal{C}=500$ random locations are chosen for the base-stations and their corresponding cells  are the Voronoi cells  of  $\mathcal{A}$. The workload demand $w_{ij}$ between cell $i$ and cell $j$ is generated uniformly at random: $w_{ij} \sim uniform(0,1)$.
\item Milano625:  We constructed this dataset from the collection of geo-referenced Call Detail Records over the city of Milan during Nov 1st, 2013 - Jan 1st, 2014 (https://dandelion.eu/). Specifically, we partition the area into a grid of $25\times 25=625$ cells of size $0.94km \times 0.94km$,  and count the calls between these cells   made during the Monday of Nov 4th, 2013.
\end{itemize}
 
\begin{figure}[t]
\centering
\includegraphics[width=0.528\textwidth]{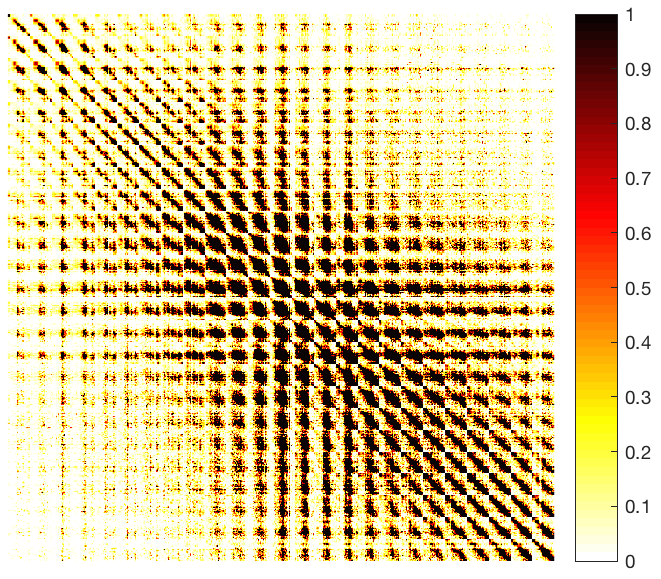}\label{fig:milano_heatmap}
\caption{Heat map of workload  $w_{ij}$ for   Milano625 , where $i$ is the x-axis, $j$ the y-axis, and $(0,0)$ at top-left.}
\end{figure}

In both studies, the transaction demand values are normalized such that they sum to 1. The number of  MEC servers  is set to $N_\mathcal{S} =10$ whose location is chosen from $N_\mathcal{L}=50$ random locations. These quantities are reasonable given the number of cells. The server capacity is set to $W\in \{0.03, 0.04, ..., 0.08\}$, meaning \{3\%, 4\%..., 8\%\} of the total workload. Figure \ref{fig:milano_heatmap} shows the heat map of the workload demand for the Milano625 dataset. 

For convenience, we refer to our three-phase algorithm as $\mathsf{KMED/FM/HUNG}$, as it applies k-median, Fiduccia-Mattheyses (FM), and  Hungarian algorithms in the three phases, respectively. Serving as benchmark  for comparison are: $\mathsf{RAND}$ (the random assignment algorithm), $\mathsf{KMED}$ (the spread-only algorithm using k-median), and $\mathsf{FM/HUNG}$ (the cost-only algorithm using FM with another step using Hungarian to improve spread). Note that the only difference between $\mathsf{KMED/FM/HUNG}$ and $\mathsf{FM/HUNG}$ is that the former starts with a    $\mathsf{KMED}$ assignment while the latter starts with a random assignment.  The metrics for comparison are cost ($COST$) and spread ($SPREAD$). $\mathsf{RAND}$ offers a good upper-bound for both cost and spread, while $\mathsf{KMED}$ represents a good lower bound (supposedly best) for spread and $\mathsf{FM/HUNG}$  a good lower-bound (supposedly best) for cost. 

The $\kappa$ parameter in Eq. (\ref{eq:facilityswap}) is set to 0.0001 for k-median and $\epsilon$ in Eq. (\ref{eq:movecells2}) is set to $\infty$ (no spread constraint  in the second phase).  The simulation runs on 10 random sets of candidate server locations and, for each set, 5 random choices for the initial assignment. The results are averaged over these 50 runs and plotted with 100\% confidence interval.

\begin{figure}[t]
\begin{center}
\subfigure[]{\includegraphics[width=0.49\textwidth]{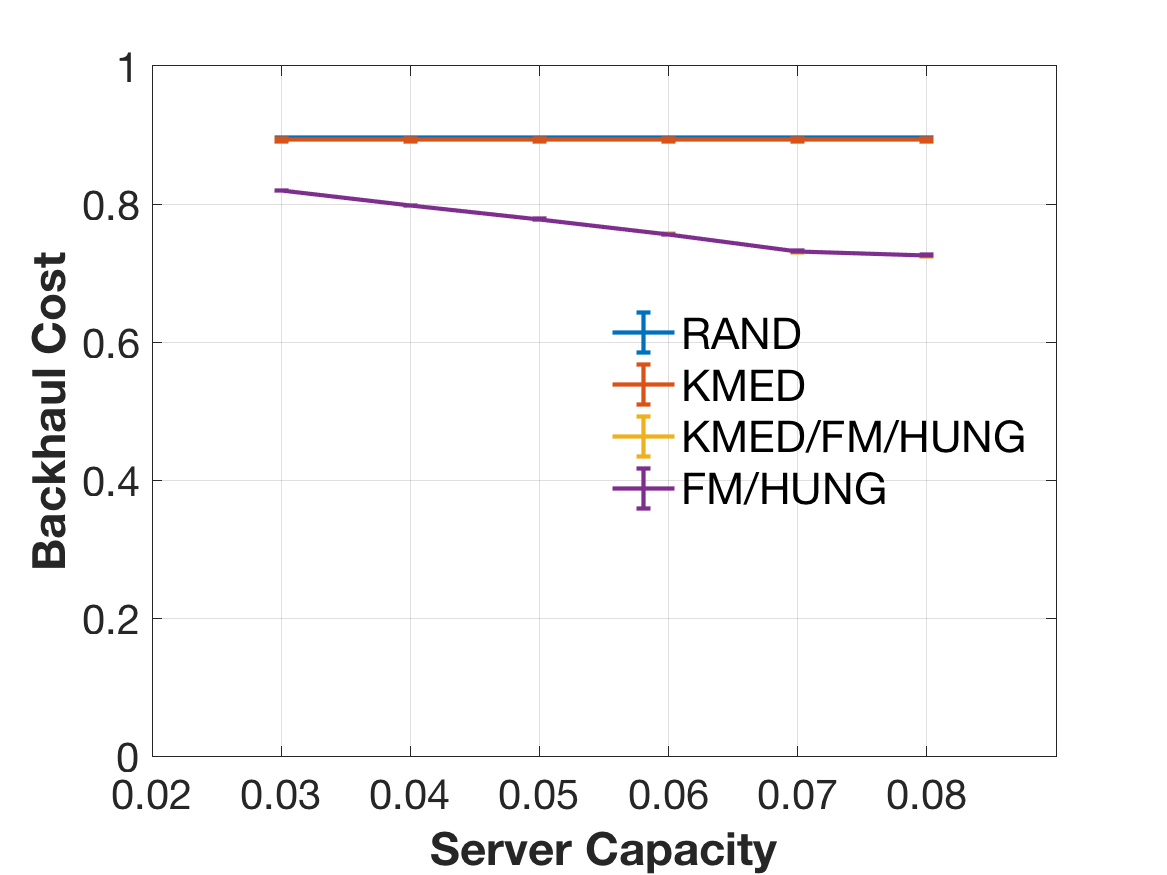}\label{fig:c500_L100_uniform_cost}}
\subfigure[]{\includegraphics[width=0.49\textwidth]{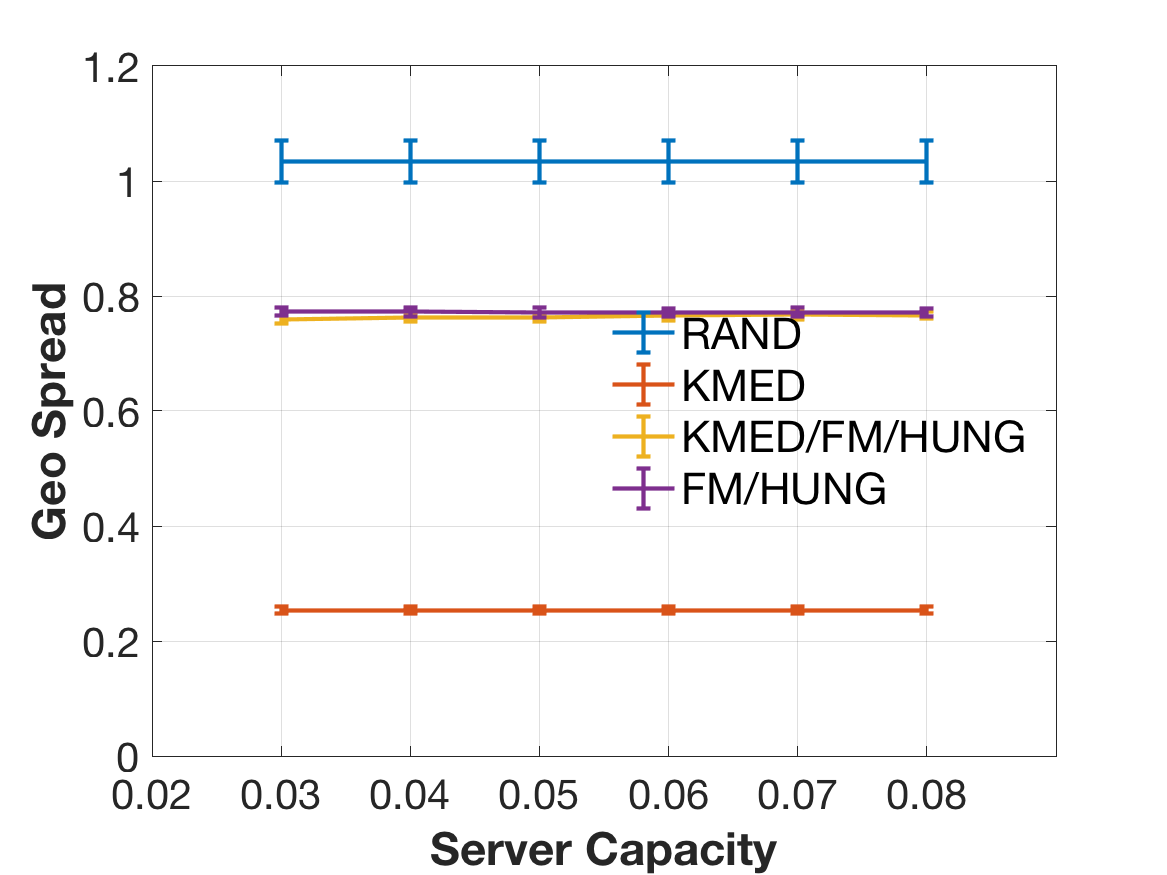}\label{fig:c500_L100_uniform_spread}}
\caption{Synthetic500: synthetic workload, no relationship with geography}
 \label{fig:c500_L100_uniform}
 \end{center}
\end{figure}
 
 \subsection{Workload Without Geography Correlation}
Figure \ref{fig:c500_L100_uniform} shows the results for the synthetic case in which geography is no factor in workload demand. In terms of cost (Figure \ref{fig:c500_L100_uniform_cost}), $\mathsf{KMED}$ is almost identical to $\mathsf{RAND}$, both incurring a high backhaul cost of almost 0.9 (i.e., 90\% of the total workload), even when the server capacity increases. This is not surprising because $\mathsf{KMED}$ is cost-blind and so when workload has no relationship with geography, minimizing spread results in a cost as bad as that of a random assignment. Perhaps for the same reason, the other two methods, $\mathsf{KMED/FM/HUNG}$ and $\mathsf{FM/HUNG}$,  also incur almost the same cost. In other words, whether we start with a $\mathsf{RAND}$ assignment or a $\mathsf{KMED}$ assignment, an application of FM+HUNG would result in similar costs. It is important to note that the FM step is effective, especially as the server capacity increases. With a server capacity of 0.08,  applying FM reduces the  backhaul cost to 0.73, a 20\% improvement from the initial assignment. 

In terms of spread (Figure \ref{fig:c500_L100_uniform_spread}), $\mathsf{KMED}$ is the best (expected) and $\mathsf{RAND}$   the worst (understandable because it is spread-blind). Between the other two,  more interestingly, $\mathsf{FM/HUNG}$ has a similar spread (only slightly larger) compared to that of $\mathsf{KMED/FM/HUNG}$. 
This study suggests  that, for a workload input that has no relationship with geography, (1) we can do better than a random assignment, (2)   the 3-phase algorithm can run without Phase 1 ($\mathsf{KMED}$) which has almost zero benefit, and (3) there is no clear winner between $\mathsf{FM/HUNG}$ and $\mathsf{KMED}$; which one should be chosen depends on whether we prefer minimizing cost or spread.

 \subsection{Workload With Geography Correlation}
\begin{figure}[t]
\begin{center}
\subfigure[]{\includegraphics[width=0.49\textwidth]{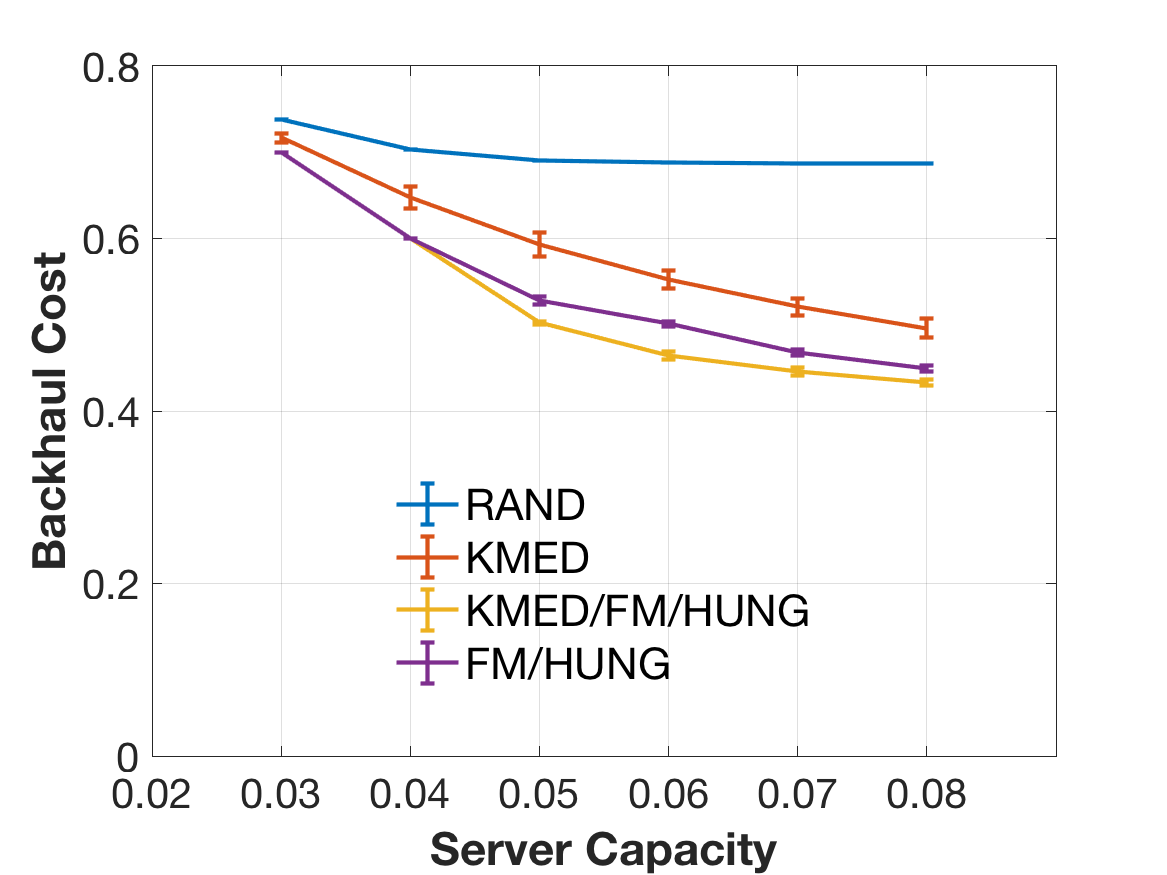}\label{fig:milano20131104_c625_L100_cost}}
\subfigure[]{\includegraphics[width=0.49\textwidth]{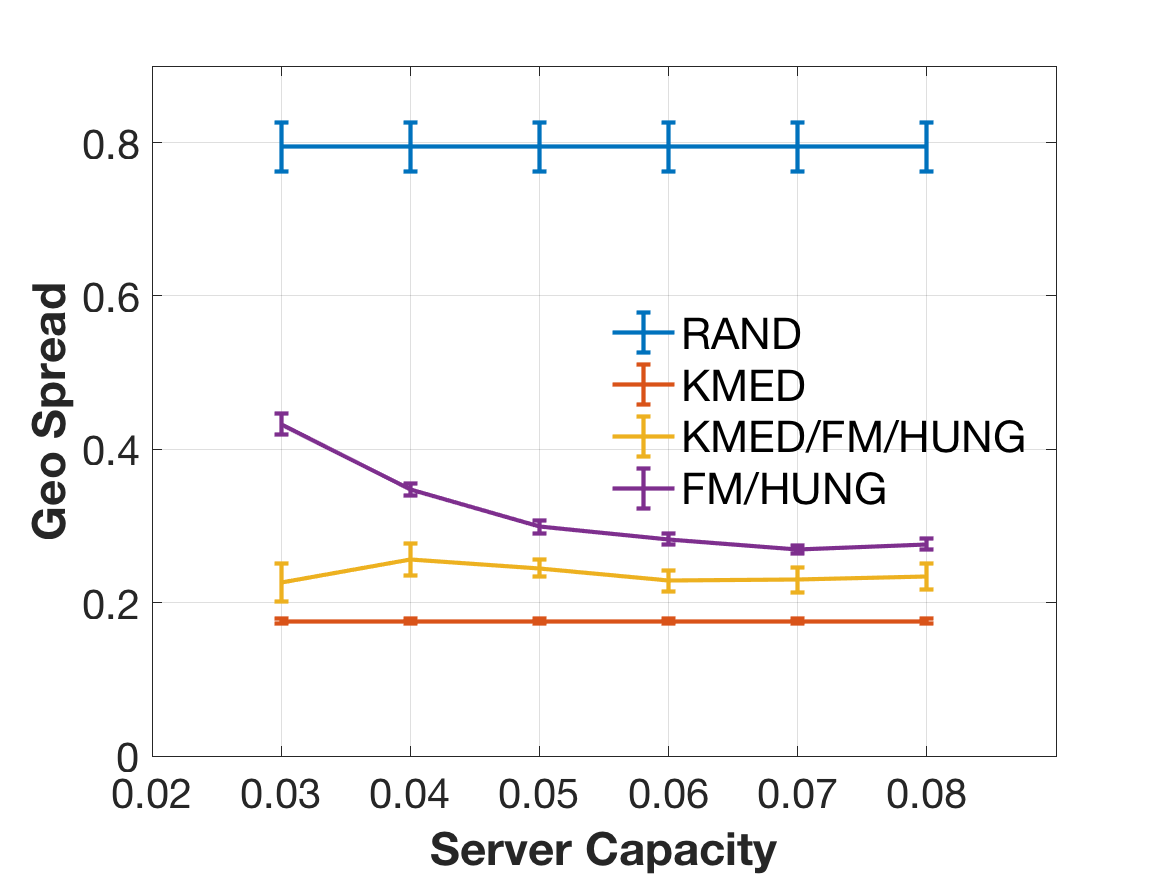}\label{fig:milano20131104_c625_L100_spread}}
\caption{Milano625: real-world workload, strongly correlated to geography.}
 \label{fig:milano20131104_c625_L100}
 \end{center}
\end{figure}
 
Figure \ref{fig:milano20131104_c625_L100} shows the results for the real-world dataset (Milano625), in which the workload has a strong correlation with geography; specifically, higher between cells of shorter distance \cite{Bouet:2017:GMR:3098208.3098216}. Similar to the above study, and, expectedly,  $\mathsf{RAND}$ is  worse than all the other algorithms in both objectives and $\mathsf{KMED}$ has the best spread. There are, however, key differences. 

First,  $\mathsf{KMED}$  has a substantially lower cost than $\mathsf{RAND}$'s; this implies that, due to the correlation between workload and geography, by minimizing spread there is, to some extent, a benefit in reducing the cost. Indeed, because $\mathsf{KMED}$ tends to cluster cells near each other  and workload demand is high between close cells, heavy workloads tend to be served by the edge, hence less workload going backhaul (compared to a random assignment). 
 Second, $\mathsf{KMED/FM/HUNG}$ is clearly better than $\mathsf{FM/HUNG}$ in both objectives; this substantiates the effectiveness of having Phase 1 ($\mathsf{KMED}$) in our  algorithm,  leading to not only better cost but also better spread. 
  Third, $\mathsf{KMED/FM/HUNG}$ has a spread only slightly worse than $\mathsf{KMED}$; this shows the effectiveness of Phase 3 ($\mathsf{HUNG}$) in improving spread. In short, all the three phases in the proposed algorithm ($\mathsf{KMED/FM/HUNG}$) are important to achieving both objectives. 

\subsection{Other Observations}
Figure \ref{fig:Milano625Visual} gives a visual representation of the assignment map according to $\mathsf{KMED}$ and $\mathsf{KMED/FM/HUNG}$.   Both methods are consistent with the physical map of Milan (Figure \ref{fig:milanomap}); that is, because most transactions involve the inner neighborhoods, a server closer to the the central area covers fewer cells (which have high activity) than those in the outskirt (which have low activity). While $\mathsf{KMED}$ places the servers spatially nicely (Figure \ref{fig:kmedmap}, Figure \ref{fig:kmedmap1}), $\mathsf{KMED/FM/HUNG}$ allows for some dis-contiguity in the cells that belong to the same server (Figure \ref{fig:kfhmap}, Figure \ref{fig:kfhmap1}); the latter does so to reduce the backhaul cost. For example, when the capacity $W=0.05$, $(COST, SPREAD)$ is (0.43, 0.21) for $\mathsf{KMED/FM/HUNG}$ and (0.49, 0.18) for $\mathsf{KMED}$. This is a tradeoff between cost versus spread. Although we cannot avoid this tradeoff, it is important to point out that  $\mathsf{KMED/FM/HUNG}$ offers a better workload balance, as clearly illustrated in Figure \ref{fig:loadbalance}. The ratio of the maximum to the minimum workload demand is at least two times less with $\mathsf{KMED/FM/HUNG}$ than with $\mathsf{KMED}$. 

\begin{figure}[t]
\begin{center}
\subfigure[$\mathsf{KMED}$]{\includegraphics[width=0.49\textwidth]{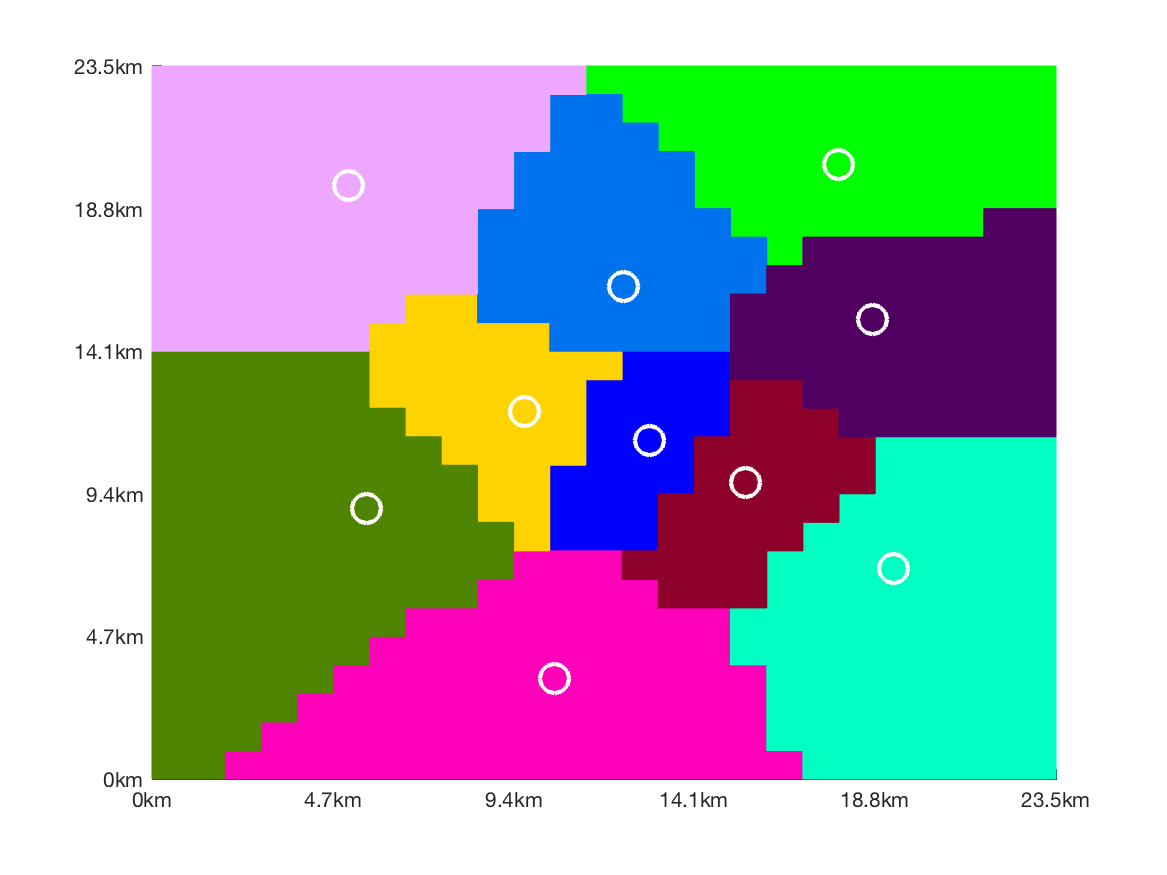}\label{fig:kmedmap}}
\subfigure[$\mathsf{KMED}$]{\includegraphics[width=0.49\textwidth]{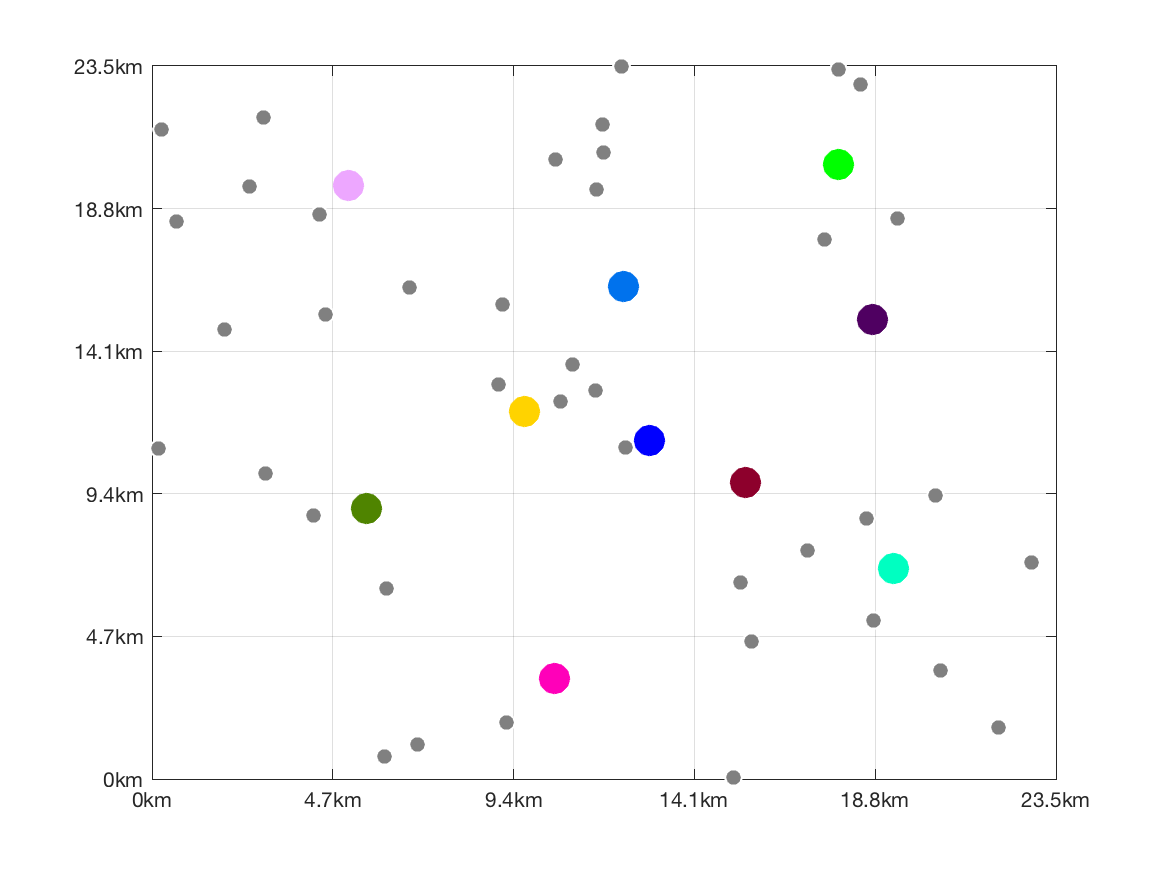}\label{fig:kmedmap1}}
\subfigure[$\mathsf{KMED/FM/HUNG}$]{\includegraphics[width=0.49\textwidth]{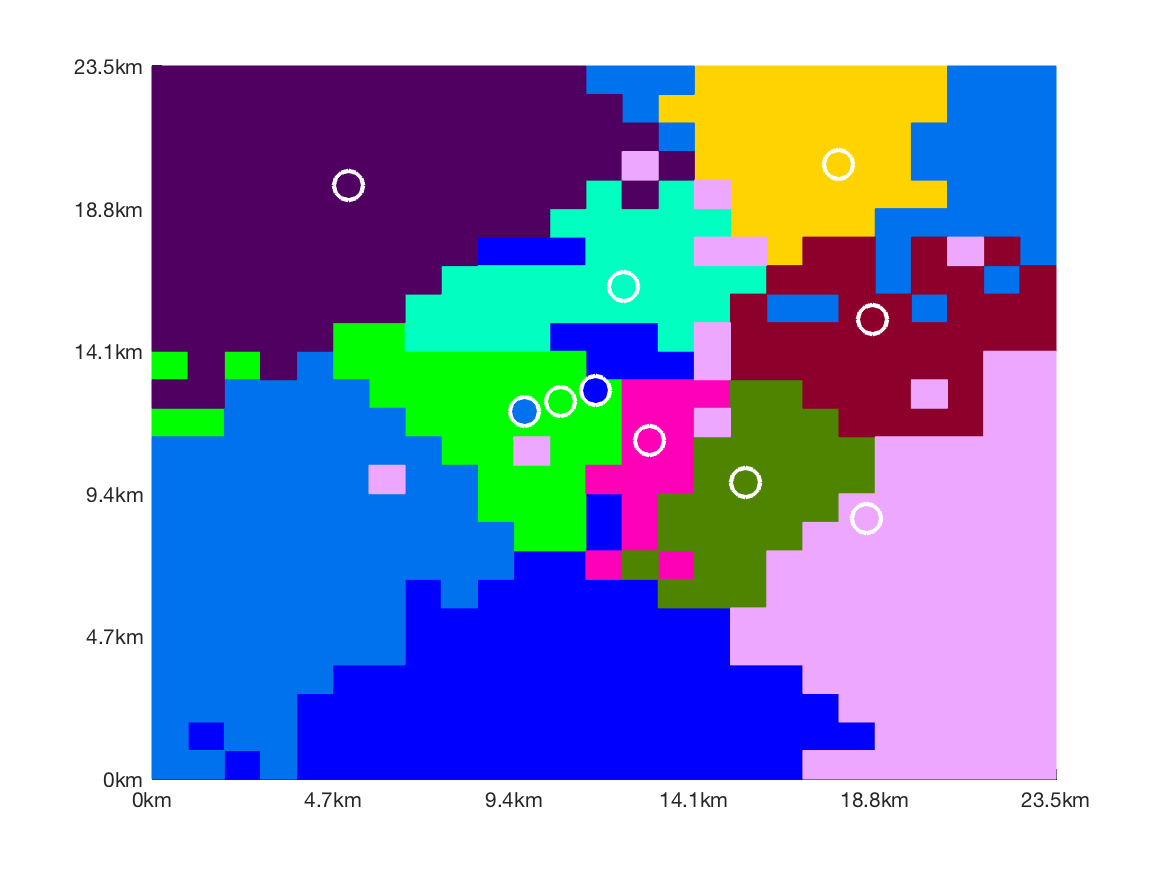}\label{fig:kfhmap}}
\subfigure[$\mathsf{KMED/FM/HUNG}$]{\includegraphics[width=0.49\textwidth]{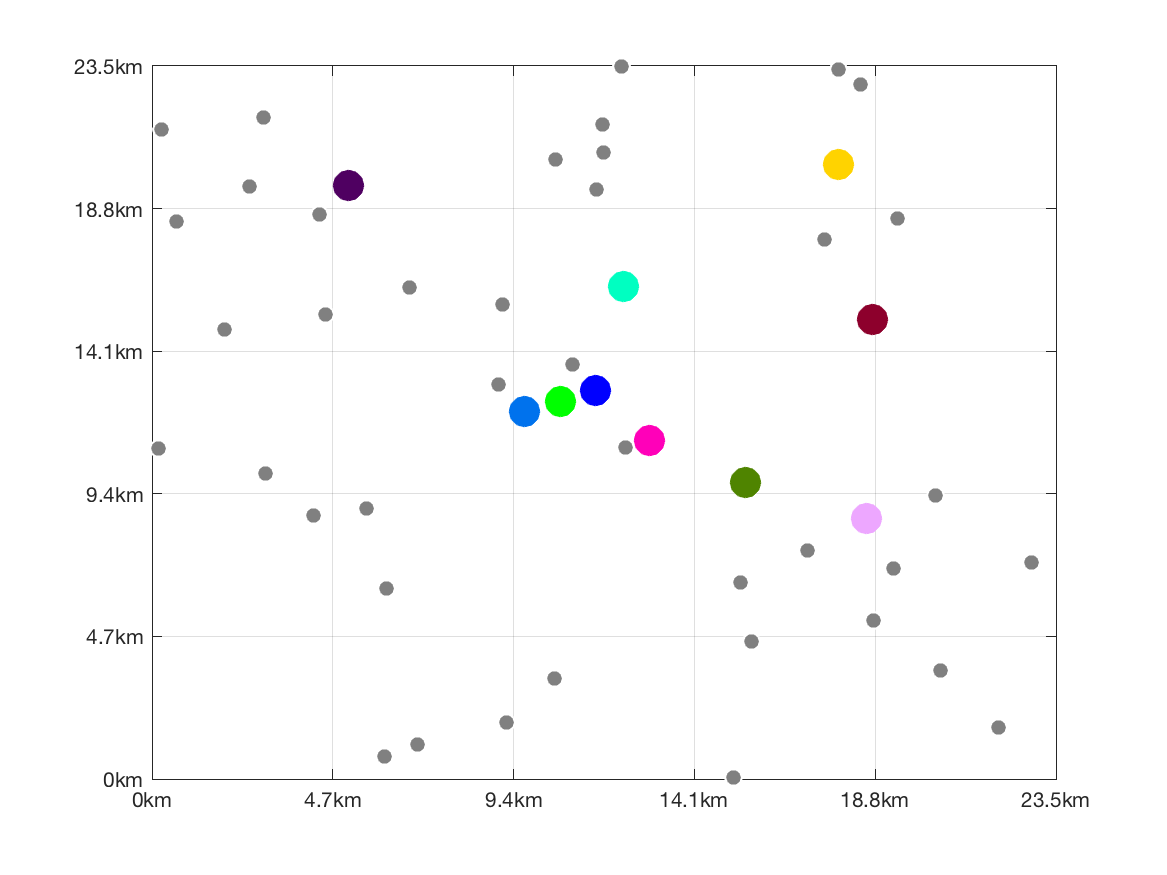}\label{fig:kfhmap1}}
\subfigure[$\mathsf{KMED/FM/HUNG}$]{\includegraphics[width=0.53\textwidth]{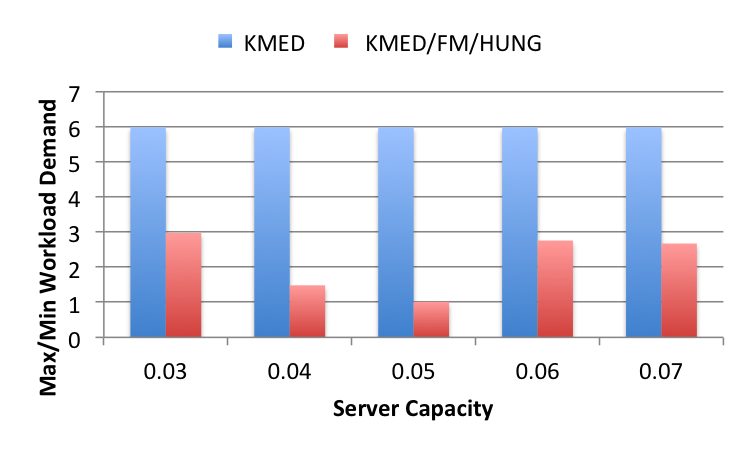}\label{fig:loadbalance}}
\subfigure[Milano area]{\includegraphics[width=0.38\textwidth,height=0.23\textheight]{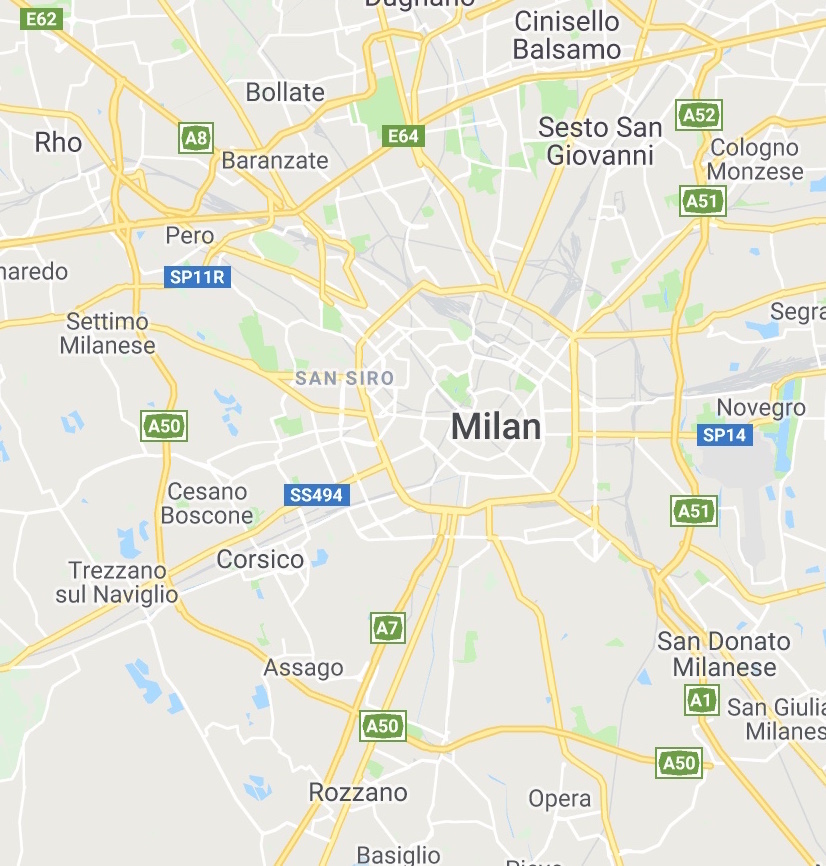}\label{fig:milanomap}}
\caption{Assignment map of (a,b) $\mathsf{KMED}$ and (c,d) $\mathsf{KMED/FM/HUNG}$ for a single random run on Milano625 with capacity $W=5\%$, where each colored-circle represents a server (out of 50 possible locations) and each server and its cells share the same color. Also shown is (e) the ratio of maximum to minimum workload and (f) the physical map of Milan.}
 \label{fig:Milano625Visual}
 \end{center}
\end{figure}

\section{Conclusions\label{sec:conclusions}}
We have addressed a new server assignment problem for MEC, whose decision making needs to be made for where geographically to place the servers and how to assign them to the user cells based on transactional workloads. The formulation of the two objectives with respect to the backhaul cost and geographic spread has not appeared in the MEC literature. We have proposed and evaluated a heuristic solution leveraging k-median, Fiduccia-Mattheyses, and Hungarian methods. The solution is not optimal (due to the NP-hardness of the problem), but an effective approximation. For the future work, our next step is to consider the case where the workload demand is not static. In practice, the workload demand varies over the time, but usually follows a pattern. Knowing this pattern, for example, in the form of a probability distribution, an interesting goal is to compute an assignment offering the best expected optimization.  

\bibliographystyle{IEEEtran}


\end{document}